\documentclass[letterpaper, 10 pt, conference]{ieeeconf}

\usepackage{amsthm,mathtools}
\usepackage{subcaption}
\usepackage{graphicx,tikz,float}
\usepackage{physics}
\usetikzlibrary{arrows,positioning,calc}
\newcommand{\enquote}[1]{``#1''}

\IEEEoverridecommandlockouts
\def\QED{\mbox{\rule[0pt]{1.3ex}{1.3ex}}}

\newcommand{\R}[1]{\mathbb{R}^{#1}}

\newcommand{\D}{\mathrm{D}}  

\newcommand{\cX}{\mathcal{X}}

\usepackage{amssymb,amsmath,color}

\let\labelindent\relax
\usepackage{enumitem}

\usepackage[english]{babel}
\usepackage{graphicx,epstopdf,float,tikz}
\usetikzlibrary{arrows,positioning,calc}

\usepackage[noadjust,compress]{cite}

\newtheoremstyle{mytheoremstyle}
{\smallskipamount}
{\smallskipamount}
{\itshape}
{}
{\bfseries}
{.}
{.5em}
{}

\theoremstyle{mytheoremstyle}

\newtheorem{assumption}{Assumption}
\newtheorem{lemma}{Lemma}

\newtheorem{corollary}{Corollary}

\newtheorem{theorem}{Theorem}
\newtheorem{definition}{Definition}

\newenvironment{syseq*}{%
	\color{gray} \left\{ \normalcolor
	\begin{aligned}
	}{%
	\end{aligned} \right. }

\newcommand{\col}{\operatorname{col}}
\newcommand{\vpre}{v_{\text{pre}}}
\newcommand{\Ek}{E_{\text{K}}}
\newcommand{\Ena}{E_{\text{Na}}}
\newcommand{\gl}{\bar{g}_{\text{leak}}}

\newcommand{\cI}{\mathcal I}

\newcommand{\cS}{\mathcal S}

\newcommand{\cV}{\mathcal V}
\newcommand{\cZ}{\mathcal Z}

\DeclareMathOperator{\card}{card}

\definecolor{dark-magenta}{RGB}{176, 70, 161}

\allowdisplaybreaks

\graphicspath{{figs/}}

\title{\LARGE \bf On the Contraction of Excitable Systems}
\author{Alessandro Cecconi$^{1,2}$, Michelangelo Bin$^{1}$, Lorenzo Marconi$^{1}$, Rodolphe Sepulchre$^{2,3}$%
	\thanks{$^{1}$ Alessandro Cecconi, Michelangelo Bin, and Lorenzo Marconi are with the Department of Electrical, Electronic and Information Engineering (DEI), University of Bologna, Bologna, Italy.}
	\thanks{$^{2}$ Alessandro Cecconi and Rodolphe Sepulchre are with STADIUS, Department of Electrical Engineering (ESAT), KU Leuven, Leuven, Belgium.}
	\thanks{$^{3}$ Rodolphe Sepulchre is with the Department of Engineering, University of Cambridge, Cambridge, United Kingdom.}%
}
\begin{document}
	
	\maketitle
	\thispagestyle{empty}
	\pagestyle{empty}
	
\begin{abstract}
	We study the contraction  of Hodgkin-Huxley model and its role in the reliability of spike timings. Without input, the model is contractive in the region of physiological interest. With impulsive synaptic inputs, contraction is retained provided that the input events are sparse enough. Contraction is lost when the input firing rate is too high. Spike timings are shown to be reliable in the contracting regime.
\end{abstract}

\section{Introduction}
Excitable systems, exemplified by conductance-based models, mix the continuous and the discrete. Their trajectories follow differential equations but consist of sequences of discrete events \cite{RS2022Spiking}. Building on this mixed description, neuromorphic control exploits excitable internal models to achieve event regulation, rather than trajectory regulation. Internal event generators can synchronize with external events despite a persistent mismatch between the corresponding continuous-time trajectories \cite{cecconi2025NOLCOS, sepulchre2025regulation}. In neuroscience, reliability refers to the trial-to-trial reproducibility of spike times under the same input stimulation, with output spike voltages that lock to input synaptic currents despite variability in initial conditions and other perturbations \cite{mainen_reliability_1995, brette_reliability_2003, lin_stimulus-response_2013}. The objective of this work is to show that an excitable behavior is contractive when regarded as an event generator, meaning that the input current keeps the state variables near an exponentially stable equilibrium except for brief spikes, that is, short and sharp excursions separated by refractory periods. In this setting, the continuous contraction property of the flow explains the discrete reliability of the observed events.

Contraction theory formalizes robustness as incremental exponential convergence of trajectories driven by the same input  \cite{lohmiller_contraction_1998, pavlov_uniform_2006, bullo_contraction_2024}. In excitable systems contraction is input dependent  \cite{pogromsky2013input, Lee2022, Bin2025reliability}. Constant or effectively constant inputs can generate limit-cycle oscillations that are not contractive and that are sensitive to phase perturbations \cite{izhikevich_dynamical_2006, Ermentrout_Terman_2010, gerstner2014neuronal}. In contrast, subthreshold inputs, or event inputs that trigger sparse spikes, result in contracting behavior with trajectories that contract to each other, and with spike times that are reliable. This input-selective viewpoint clarifies how the same neuron can appear phase-sensitive in a tonic regime yet display highly reproducible timing when driven by sparse triggering signals.

We formalize this viewpoint on a standard conductance-based neuron model driven by a first-order synapse. Physiological bounds on voltage and gating variables define a compact forward-invariant set on which the unforced Hodgkin–Huxley flow is incrementally exponentially stable. To connect the continuous flow with discrete observations we pair the model with an event readout that extracts spike timings from a voltage trajectory. When the flow is contractive, different initial conditions under the same admissible input lead to trajectories that converge and to event sequences that coincide after a short transient.

We then study impulsive synaptic drives and identify conditions under which contraction, and therefore event-level reliability, persists. For arbitrary spike trains, an average dwell-time condition ensures impulses are sufficiently spaced so that the state returns near equilibrium between events. For periodic spike trains, this requirement becomes a lower bound on the inter-impulse period that depends on the contraction margin of the unforced flow. The opposite regime reveals the limits of reliability. At high input rates the synaptic state saturates, the effective conductance becomes nearly constant, and the neuron may enter self-sustained periodic firing on a limit cycle, which destroys global contractivity and makes timing sensitive to phase.

Numerical experiments on a Hodgkin–Huxley neuron illustrate both behaviors predicted by the analysis. With adequately spaced inputs, trajectories contract and spike times lock tightly after a brief transient. With high-rate input events, tonic oscillations exhibit persistent phase offsets across trials and timing variability does not shrink. Together these findings provide a simple and testable account of when excitable systems are contractive at the level of events and why their spike timings are reliable.

	\section{Modeling Excitable Systems}
	We consider excitable systems modeled by the following \emph{conductance-based model}
	\begin{equation}\label{eq:exc_model}
		\begin{aligned}
			\tau_j(v)\dot{x}_j &= -x_j + \mu_j(v), \qquad j=1,\dots,m,\\
			C\dot{v} &= -\bar g_{\mathrm{leak}}(v-E_{\mathrm{leak}})
			-\sum_{k=1}^n i_k(x,v) - i_s,
		\end{aligned}
	\end{equation}
	where $v\in\mathbb R$ is the membrane voltage, and $x_j\in\mathbb R$ are gating variables.
	The functions $\mu_j:\mathbb R\to(0,1)$ and $\tau_j:\mathbb R\to\mathbb R_+$ are the steady-state activation and the voltage-dependent time constant, respectively. These functions typically have a sigmoidal shape \cite[Sec.~1.8]{Ermentrout_Terman_2010}.
	
	The quantities $i_k(x, v)$ represent internal ionic currents. They have an Ohmic form
	\[
	i_k(x,v)=\bar g_k\,\varphi_k(x)(v-E_k),\qquad k=1,\dots,n,
	\]
	where $\bar g_k>0$ is the maximal conductance, $E_k\in\mathbb R$ the corresponding reversal potential, and $\varphi_k:\mathbb R^m\to(0,1)$ is an activation function determining the fraction of open channels. 
	
The variable $i_s$ denotes an external synaptic current, which is modeled by
\begin{align}\label{eq.synapse}
	\dot s = - \alpha s + \beta (1-s)\vpre(t), \quad
	i_s(t) = \bar g_s s(v - E_s),
\end{align}
where $\alpha > 0$ is a relaxation constant, $\beta \in(0,1]$ is the event-activation gain, $\bar g_s>0$ is the maximal synaptic conductance, $E_s\in\R{}$ is the synaptic reversal potential, and $\vpre(t)$ is the pre-synaptic voltage, the model input. The synapse is a one-port element of the same type as the internal currents.

We model presynaptic spikes as a locally finite set of event times. Let
$\mathbb U$ denote the collection of all locally finite subsets of
$\mathbb R_{\ge 0}$. Each $\mathcal I\in\mathbb U$ can be uniquely
represented by a strictly increasing sequence
$
\mathcal I=\{t_\ell\}_{\ell=1}^{N},
$
where $N\in\mathbb N$. To every $\mathcal I$ we associate
the spike-train signal
\begin{equation}\label{eq.disc.v_pre.dirac}
	\vpre(t)\coloneq \sum_{t_\ell \in \mathcal I} \delta(t-t_\ell),
\end{equation}
where $\delta$ denotes the Dirac distribution. In the particular case of
$t_\ell=\ell T$, the spike train is $T$-periodic. This representation is
standard in neuroscience and neuromorphic
engineering~\cite{gerstner2014neuronal,Medvedeva2025}.

The Hodgkin--Huxley model~\cite{Hodgkin1948} is recovered as a
particular case with $m=3$ gating variables modulating $n=2$ ionic
currents,
\[
i_{\mathrm{Na}}=\bar g_{\mathrm{Na}}\,x_1^{3}x_2(v-E_{\mathrm{Na}}),
\qquad
i_{\mathrm K}=\bar g_{\mathrm K}\,x_3^{4}(v-E_{\mathrm K}),
\]
where $i_{\mathrm{Na}}$ and $i_{\mathrm K}$ denote the sodium and
potassium currents, respectively. The corresponding equivalent circuit
is shown in Figure~\ref{fig.circuit}.
	
	\begin{figure}
		\centering
		\includegraphics[width=0.475\textwidth]{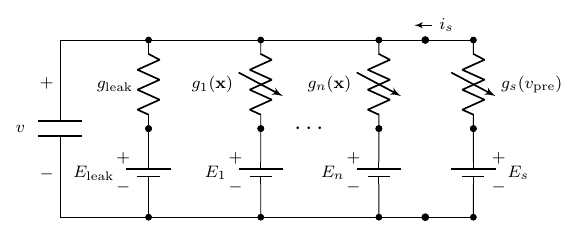}
		\caption{Circuit schematic of the conductance-based model \eqref{eq:exc_model} with the external synapse \eqref{eq.synapse}.}
		\label{fig.circuit}
	\end{figure}
	
	In physiological conditions, neurons have a well-defined resting state, which corresponds to a unique stable equilibrium. Furthermore, the membrane voltage remains confined in an interval determined by the minimum and maximum of the reverse potentials, that is, in the case of the Hodgkin-Huxley, the Nernst potential of potassium and the Nernst potential of sodium, respectively \cite{izhikevich_dynamical_2006}. We link such conditions to our model \eqref{eq:exc_model}-\eqref{eq.synapse}.
	With $\vpre\equiv 0$, the unforced voltage dynamics in \eqref{eq:exc_model} can be rewritten as
	\begin{equation}
		C\dot v=-G(x,s)\big(v-E(x,s)\big),
	\end{equation}
	where
	\begin{equation}
		G(x,s)\coloneq \bar g_{\mathrm{leak}}+\sum_{k=1}^n \bar g_k \varphi_k(x)+ \bar g_s s \ge \gl > 0,
	\end{equation}
	and
	\begin{equation}
		E( x,s)\coloneq
		\frac{\bar g_{\mathrm{leak}}E_{\mathrm{leak}}+\sum_k \bar g_k \varphi_k(x) E_k + \bar g_s s\,E_s}
		{\bar g_{\mathrm{leak}}+\sum_k \bar g_k \varphi_k( x)+ \bar g_s s}.
	\end{equation}
	Hence, $E({x}, s)$ is a convex combination of reversal potentials. In particular, let $E_\text{min} \coloneq \min\{E_{\mathrm{leak}},E_k,E_s\}$, and $E_{\text{max}} \coloneq \max\{E_{\mathrm{leak}},E_k,E_s\}$. Then, for every $(x, v)$, it holds
	\[
	E( x,s)\in [E_{\text{min}}, E_{\text{max}}].
	\]
	For instance, for the Hodgkin–Huxley model one has $E({x}, s) \in [\Ek, \Ena]$, provided that $E_\text{s} \in (\Ek, \Ena)$. More in general, one can show the following.
	\begin{lemma}
		With $\cS=[0,1]$, $\cX=[0,1]^m$, $\cV=[E_{\min},E_{\max}]$, define $\cZ\coloneq \cS\times\cX\times\cV$. Then, $\cZ$ is compact and forward invariant for \eqref{eq:exc_model}–\eqref{eq.synapse} under \eqref{eq.disc.v_pre.dirac}.
		\label{lemma.fi}
	\end{lemma}
	\begin{proof}
		See Appendix.
	\end{proof}
	As for the rest state, we notice that an equilibrium $(s^*, x^*,v^*)$ of \eqref{eq:exc_model}–\eqref{eq.synapse} with $\vpre\equiv 0$, necessarily satisfies $s^*=0$, and the current balance
	\begin{equation}\label{eq.current.balance}
		\bar g_{\mathrm{leak}}(v^*-E_{\mathrm{leak}})+\sum_{k=1}^n \bar g_k \varphi_k( x^*)(v^*-E_k)=0.
	\end{equation}
	Assuming, as customary, $v^*=0$ (which can be done without loss of generality by appropriately re-defining the reverse potentials), we obtain $x_j^*=\mu_j(0)$, and \eqref{eq.current.balance} reduces to
	\begin{equation}
		\bar g_{\mathrm{leak}}E_{\mathrm{leak}}+\sum_{k=1}^n \bar g_k \varphi_k(x^*) E_k=0,
	\end{equation}
	Namely, $\bar g_{\mathrm{leak}}E_{\mathrm{leak}}$ compensates for the residual reversal
	potentials at $x^*$, ensuring that $(s^*, x^*,v^*) = (0, x^*, 0)$ is an equilibrium for \eqref{eq:exc_model}-\eqref{eq.synapse}.
	
	For the sake of compactness, we rewrite \eqref{eq:exc_model}–\eqref{eq.synapse} as
	\begin{equation}\label{eq.exc.compact}
		\dot z = F(z) + G(z)\,u,\quad z(0)=z_0,
	\end{equation}
	where $z \coloneq \col(s,x,v)\in\cZ$, $u \coloneq \vpre$, and
	\begin{align*}
		F(z) &\coloneq \col\big(-\alpha s,\ f(z)+g(z)s\big), \\
		G(z) &\coloneq \col\big(\beta(1-s),\ \mathbf{0}_{m + 1}\big), \\
		g(z)&\coloneq \col\big(\mathbf{0}_{m},\ -(\bar g_s/C)(v-E_s)\big),
	\end{align*}
	and where $f(z)$ collects the intrinsic unforced conductance–based dynamics of ${x}$ and $v$ in \eqref{eq:exc_model}. 
	
	While the previous conditions do not by themselves imply uniqueness of the rest state, assuming a unique attractive and locally exponentially stable equilibrium is both standard in the literature and well supported empirically for conductance-based neuron models \cite{izhikevich_dynamical_2006}. We therefore adopt the following.
	\begin{assumption}\label{ass.les.gas}
		$z^* \coloneq (0, x^*, 0)$ is the unique equilibrium of \eqref{eq.exc.compact} with $u \equiv 0$, and it is locally exponentially stable (LES) with a domain of attraction including $\cZ$.
	\end{assumption}
	 Indeed, for \eqref{eq.exc.compact} without a constant bias ($ u(t) \not\equiv\bar u$), the resting equilibrium is typically unique. When a constant current is introduced and tuned past a critical value, the rest state either loses stability or disappears via a bifurcation, and a stable periodic orbit becomes the attracting behavior \cite{izhikevich_neural_2000}.
	Under Assumption \ref{ass.les.gas}, $z^*$ is also exponentially stable on $\cZ$, as established by the following lemma.
	\begin{lemma}\label{lemma.les.attractive.es}
		If an equilibrium $z^*\in\cZ$ for \eqref{eq.exc.compact} with $u \equiv 0$ is locally exponentially stable and attractive on a compact forward-invariant set $\cZ$, then $z^*$ is exponentially stable on $\cZ$.
	\end{lemma}
	\begin{proof}
		See Appendix.
	\end{proof}
	
	We now associate with \eqref{eq.exc.compact} a discrete-event model by composing its continuous-time input--output dynamics with an event readout.
	
Given $z_0\in\cZ$ and $\mathcal I\in\mathbb U$, let $\varphi(t,z_0,\mathcal I)$ denote the solution of \eqref{eq.exc.compact} driven by \eqref{eq.disc.v_pre.dirac}, and let $v(t,z_0,\mathcal I)$ be its voltage component. We introduce an \emph{event readout} as a map
\begin{equation}\label{eq.discrete.readout}
\mathcal R:\mathcal V\to\mathbb U,
\end{equation}
where $\mathcal V \subset C^1(\R{}_+)$ is a suitable class of voltage signals for which the detection rule is well posed and produces a locally finite event train. For each $v\in\mathcal V$, the set $\mathcal R[v]$ is the event train extracted from the voltage trajectory according to a prescribed detection rule. As an example, one may define $\mathcal R$ by upward threshold crossings with minimum velocity
\[
t_h\in\mathcal R[v]
\quad\Longleftrightarrow\quad
v(t_h)=v_{\rm th},
\quad
\dot v(t_h)\ge \gamma>0,
\]
for some threshold $v_{\rm th}$ and slope bound $\gamma>0$, together with the assumption that such crossings are isolated. This guarantees that $\mathcal R[v]$ is a locally finite set of event times.

The discrete-event model associated with \eqref{eq.exc.compact} is then defined by
\begin{equation}\label{eq.discrete.model}
	\mathcal F:\cZ\times\mathbb U\to\mathbb U,
	\quad
	\mathcal F(z_0,\mathcal I)
	:=
	\mathcal R\big[v(\cdot,z_0,\mathcal I)\big],
\end{equation}
for all pairs $(z_0,\mathcal I)$ such that $v(\cdot,z_0,\mathcal I)\in\mathcal V$. That is, for each initial condition $z_0$ and input event train $\mathcal I$, the system generates a continuous-time output trajectory $v(\cdot,z_0,\mathcal I)$, and the readout map $\mathcal R$ extracts from it the corresponding event train.
	
	A block diagram of the overall architecture is shown in Figure \ref{fig.architecture}.
	In the following, we will say that the discrete-event system is \emph{reliable} if the underlying forced continuous dynamics is contractive. This inherits the robustness properties implied by contraction, such as robustness to disturbances and parameter uncertainties.
	
	\begin{figure}
		\centering
		\includegraphics[width=0.5\textwidth]{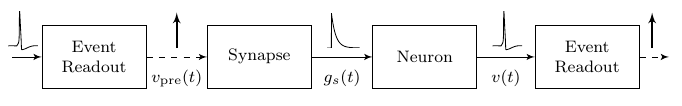}
		\caption{Block diagram of the overall architecture. Spikes are converted into event timings through an event readout. Bold lines represent continuous-time signals, while dashed lines represent event signals.}
		\label{fig.architecture}
	\end{figure}
	\section{Contraction Properties of Excitable Systems}\label{sec.contraction}
	In this section we show the main results of the paper, proving contraction of the unforced system \eqref{eq.exc.compact}, and showing how the property is maintained by certain inputs of interest.
	%
	\begin{definition}[Incremental Exponential Stability]\label{def.is}
		Let $\cZ$ be a compact forward invariant set for \eqref{eq.exc.compact}. The system \eqref{eq.exc.compact} with $u \equiv 0$ is said to be incrementally exponentially stable (IES) on $\cZ$ if there exist $k \ge 1$, and $\lambda > 0$ such that, for all $z_1,z_2\in\cZ$,
		\[
		\|\varphi(t,z_1)-\varphi(t,z_2)\|\le k e^{- \lambda t}\|z_1-z_2\| \quad \forall t\ge 0.
		\]
	\end{definition}
	It is known that incremental stability and uniform convergence coincide on compact sets \cite{RUFFER2013277}. Therefore, if system \eqref{eq.exc.compact} is IES on $\cZ$ with $u \equiv 0$, it posses a unique exponentially stable steady-state solution in $\cZ$. In our case, it coincides with the unique exponentially stable equilibrium. The next statement upgrades the pointwise stability to an incremental estimate on the whole set, hence to contraction of the unforced flow.
	
	\begin{theorem}\label{thm.contr.no.u}
		Under Assumption \ref{ass.les.gas}, consider \eqref{eq.exc.compact} with $u\equiv 0$ on the forward-invariant set $\cZ$. Let $z^*\in\cZ$ be the unique equilibrium of \eqref{eq.exc.compact}, exponentially stable on $\cZ$. Then the unforced system \eqref{eq.exc.compact} is incrementally exponentially stable (IES) on $\cZ$.
	\end{theorem}
	\begin{proof}
		See Appendix.
	\end{proof}
	
	The previous theorem is an immediate consequence of exponential stability of the equilibrium point on a compact invariant set. Since excitable systems are not necessarily uniformly contractive with respect to the input \cite{jouffroy2003simple, Bin2025reliability}, contraction may not be preserved if an input is applied. However, as the next theorem will show, for the class of inputs \eqref{eq.disc.v_pre.dirac} contraction is preserved, provided that the events in the train are, on average, sufficiently spread in time. The theorem builds on classical average dwell-time arguments \cite{Hespana1999dwell,liberzon2003switching}.
\begin{theorem}[IES under average dwell time]\label{thm:IES-dwell}
	By Theorem~\ref{thm.contr.no.u}, there exist $k> 1$ and $\lambda>0$ such that, for all $t_2\ge t_1\ge 0$ and all $z_1,z_2\in\cZ$,
	\[
	\|\varphi(t_2,z_1)-\varphi(t_2,z_2)\|
	\le k\,e^{-\lambda(t_2-t_1)}\,\|\varphi(t_1,z_1)-\varphi(t_1,z_2)\|.
	\]
	Let $u$ be the impulse train \eqref{eq.disc.v_pre.dirac}. If the impulse times satisfy the average dwell-time bound 
	\[
	N(t_2,t_1)\le N_0+\frac{t_2-t_1}{\tau_a} , \quad \forall t_2 \ge t_1 \ge 0,
	\]
	with $N_0 \ge 0$ and $\tau_a > 0$ independent of $t_2,t_1$, then the forced system \eqref{eq.exc.compact} under \eqref{eq.disc.v_pre.dirac} is IES on $\cZ$ whenever 
	\[
	\lambda>\frac{\ln k}{\tau_a}.
	\]
\end{theorem}
	\begin{proof}
		See Appendix.
	\end{proof}
	The constant $N_0$, and $\tau_a$ have the standard average dwell-time interpretation. The quantity $\tau_a$ represents a lower bound on the average time between impulses, whereas $N_0$ is a chatter bound allowing a finite number of extra impulses over short intervals. 
	
	Considering now the class of periodic signals, it is of particular interest to show how contraction depends on the signal period $T$.
	\begin{corollary}[Periodic impulse trains]\label{cor:periodic-IES}
		If the impulse train \eqref{eq.disc.v_pre.dirac} is periodic with period $T>0$, then the forced system is IES whenever
		\[
		\frac{1}{T}<\frac{\lambda}{\ln k}.
		\]
	\end{corollary}
	\begin{proof}
		See Appendix.
	\end{proof}
	
	For a periodic train with inter-impulse period $T$, if the rate of arrivals of the impulses is below the margin $\lambda/\ln k$, the exponential convergence between any two trajectories of the forced system survives. As $T$ decreases, the bound eventually fails, matching the intuition that very high-rate pulsing undermines contractivity, since the synapse tends to be always open, as proved in the following lemma.
	
\begin{lemma}[High–rate periodic impulses keep the synapse open]
	\label{lem:high-rate-synapse}
	Let \eqref{eq.synapse} be driven by \eqref{eq.disc.v_pre.dirac} with impulse times
	$t_\ell=\ell T$, $T>0$. For each $T>0$ there exists a unique $T$-periodic solution
	$s_*^{(T)}$ of \eqref{eq.synapse} forced by \eqref{eq.disc.v_pre.dirac}, which is
	globally exponentially attractive. Moreover, for every $L>0$ and every
	$\phi\in \mathcal L^1([0,L])$,
	\[
	\lim_{T \to 0}\ \int_0^L \phi(t)\,\big(s_*^{(T)}(t)-1\big)\,dt \;=\; 0,
	\]
	that is, $s_*^{(T)} \to 1$ weakly on compact intervals as $T \to 0$.
\end{lemma}
	Lemma~\ref{lem:high-rate-synapse} shows that, for $T \to 0$ the synaptic state approaches a constant value $1$. The intuition behind is that, for impulses arriving at very high rate, system \eqref{eq.exc.compact} behaves, on compact intervals of time, in the same way as it were subject to a constant input $\vpre(t) \equiv \bar{v}_{\text{pre}}$. Thus, in view of \cite[Thm. 1]{sontag2013mathematical}, this \enquote{tonic limit} destabilizes the rest state and creates a limit cycle, explaining the loss of the contraction margin in Theorem~\ref{thm:IES-dwell} as $T\to 0$. 
	
Simulations illustrating these results are shown in Figure~\ref{fig.contraction}. When impulses are too close to each other, they approximate a constant current that can induce a limit cycle. Trajectories then do not converge to a common steady-state response, the phase shift from different initial conditions persists, and the behavior is not reliable. By contrast, sufficiently sparse periodic impulses yield a unique and attractive steady-state response, so that all trajectories \emph{entrain} to the input \cite{russo2010global}, which implies reliable behavior. These observations align with \cite{Bin2025reliability}, where contraction of the forced system is linked to an average-time condition stating that trajectories must spend enough time in a contractive region to synchronize asymptotically.
	
	\begin{figure}
		\centering
		\includegraphics[width=0.5\textwidth]{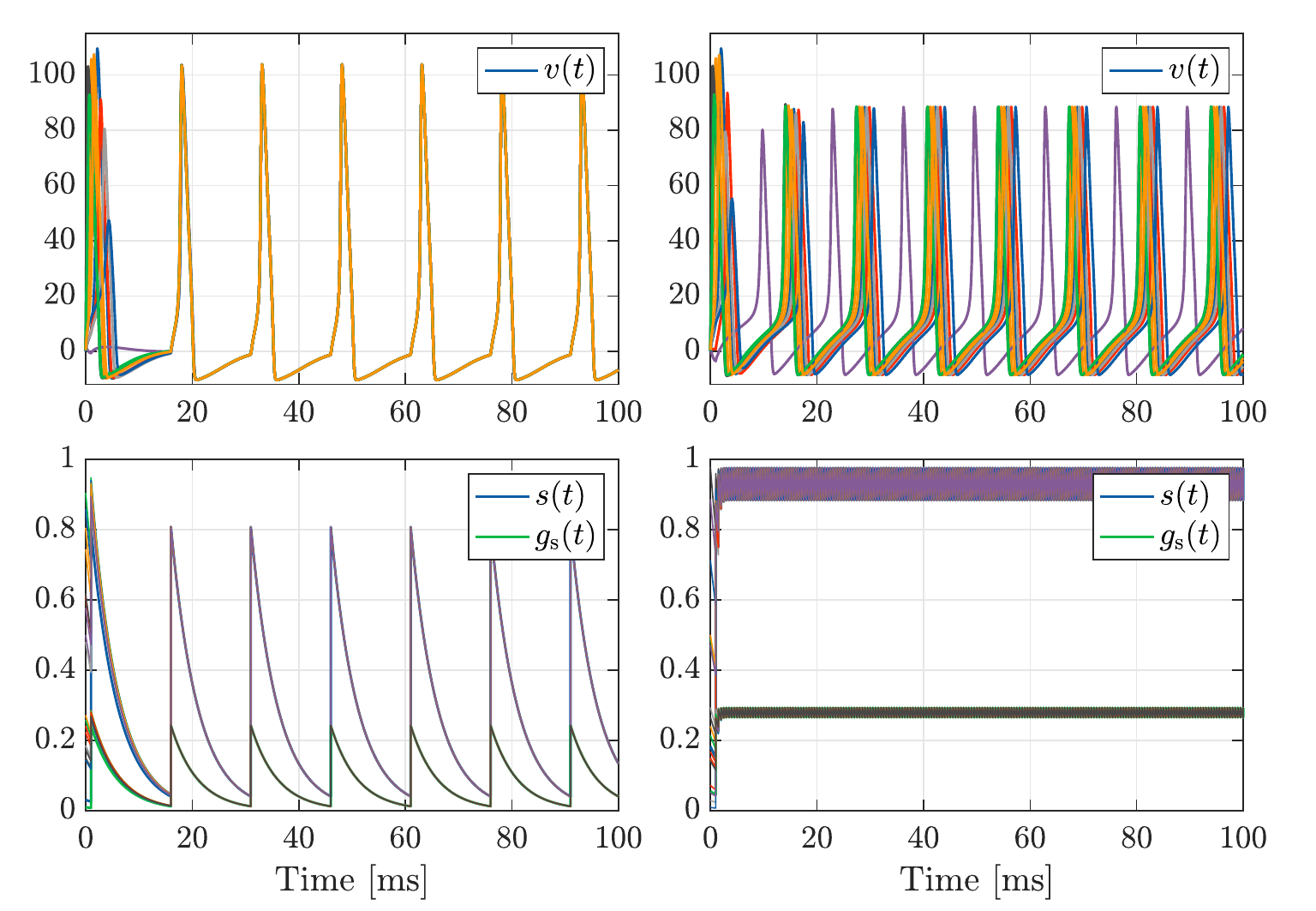}
		\caption{Hodgkin–Huxley neuron driven by periodic impulse trains. Top: membrane voltage for multiple initial conditions under the same input. Bottom: synaptic state $s(t)$ and conductance $g_{\mathrm s}(t)$. Left column (sparse, $T=15$\,ms): contraction is preserved and all trajectories converge to a unique steady-state response. Right column (dense, $T=0.5$\,ms): the synapse is effectively always open (constant input), a limit cycle emerges, and phase offsets persist across initial conditions. ($\alpha=0.8$, $\tau_s=5$\,ms, $\bar g_s=0.3$\,mS/cm$^2$, $E_{\text{s}}=65$\,mV).}
		\label{fig.contraction}
	\end{figure}
	
	\section{Event Reliability}\label{sec:reliability}
	Reliability is a property of the discrete event map, yet in our setting it is inherited directly from contraction of the continuous dynamics \cite{Bin2025reliability}. If \eqref{eq.exc.compact} is contractive on the compact forward–invariant set $\cZ$, then for any fixed input sequence $\cI\in\mathbb U$ and fixed parameters the corresponding trajectories $\varphi(t,z_0,\cI)$ converge to the same steady–state response. By means of the event readout \eqref{eq.discrete.readout}, the events extracted from these trajectories align asymptotically. In this sense, contraction of the flow implies reliability of spike times.
	
	Moreover, the robustness inherited by contraction carries over to events. Under bounded disturbances or small parameter variations across trials, incremental exponential stability yields uniform trajectory bounds that scale with the perturbation magnitude. Consequently, the induced timing errors remain uniformly \emph{bounded} after a transient. In contrast, when the flow is not contractive, differences in initial conditions or parameters can lead to persistent phase offsets and highly variable spike timings.
	
	In practice, the threshold $v_{\text{th}}$ of the event readout \eqref{eq.discrete.readout} is chosen inside the spike excursion where the upstroke is monotone, and a value near $(E_{\min}+E_{\max})/2$ often suffices. In planar excitable systems (e.g., FitzHugh–Nagumo or Morris–Lecar \cite{fitzhugh_impulses_1961, morris1981voltage}), this corresponds to placing $v_{\text{th}}$ just above the upper knee of the fast–variable nullcline, where the trajectory jumps to the excited branch and crosses the threshold exactly once with positive slope.
	
	A simulation illustrating these points is shown in Figure~\ref{fig.raster.plot}. The neuron is driven by the same random impulse train across trials while initial conditions, and parameters are perturbed. When contraction of the underlying dynamics is preserved, trajectories remain uniformly close after a transient despite these perturbations, and the extracted event sequences align up to a small bounded timing error. In the non–contractive regime, by contrast, trial–to–trial phase differences persist and spike times vary consistently.
	
	\begin{figure}[h]
		\centering
		\includegraphics[width=0.5\textwidth]{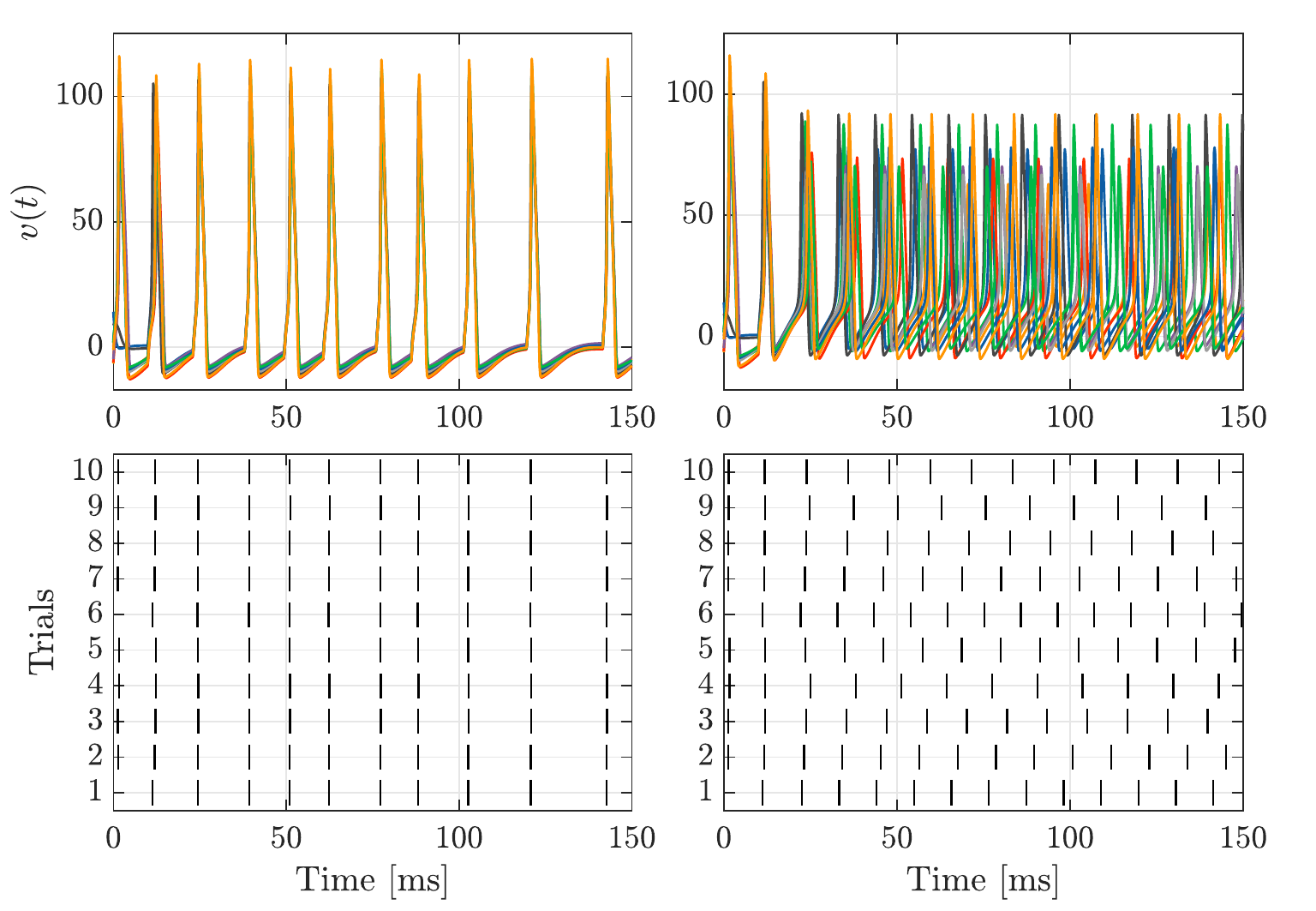}
		\caption{Hodgkin–Huxley neuron with a conductance synapse driven by two impulse trains. 
			Top: membrane voltage from $10$ trials under the same input, with random initial condition and parameters perturbed by $\pm 20\%$ of baseline values. Bottom: spike raster plots. 
			Left column (sparse): a random train with dead-time yields tightly aligned spike times across trials. Right column (dense): a uniform train with period $T=0.01$\,ms keeps the synapse effectively open, producing an almost constant conductance, with trajectories showing tonic-like activity with trial-dependent phase.
			(Baseline parameters: $\alpha=1$, $\tau_s=4$\,ms, $\bar{g}_{\mathrm s}=0.425$\,mS/cm$^2$, $E_{\mathrm s}=65$\,mV).}
		\label{fig.raster.plot}
	\end{figure}
	\section{Discussion}
\subsection{Rate vs spike code}
A long-standing distinction in neural coding contrasts spike-code and rate-code views \cite{brette_philosophy_2015}. In a spike-based view, precise event times are taken to be the meaningful carrier of information. In a rate-code view, by contrast, the firing rate is regarded as the fundamental signal, effectively averaging over spikes and discarding timing precision.
Our analysis provides a dynamical perspective on this distinction. The key point is that precise event times can serve as meaningful signals only when they are reproducible across repeated trials with the same input, up to small perturbations. This is exactly the kind of robustness promoted by contraction. When the forced flow is contractive, differences in initial conditions are forgotten exponentially fast, and the output spike train becomes a stable feature of the input-output response. On the other hand, when contraction is lost, the situation changes qualitatively. The same input may produce responses that remain offset in phase or exhibit persistent timing discrepancies across trials. In such a regime, the exact timing of each spike is no longer robust, and interpreting the event train itself as the signal becomes questionable. What may still remain reproducible is a coarser observable, such as the average number of spikes over a time window. This naturally favors a rate-based description.
From this viewpoint, the distinction between spike code and rate code is not only a matter of representation, but also a matter of dynamical regime. Contractive responses support a spike-based description because the spike pattern is a stable attractor of the dynamics. Non-contractive responses, by contrast, undermine the reliability of individual event times and therefore motivate a rate-level description. In this sense, contraction provides a criterion for when an excitable system can be regarded as a reliable event generator, and when only a coarse measure of activity should be trusted.
	
\subsection{Endogenous vs  exogenous oscillators}
Mathematical models of rhythm generation often fall into two distinct categories: endogenous oscillators and exogenous oscillators. Endogenous oscillators are modelled as limit cycle oscillations of autonomous differential equations \cite{Pikovsky2001}. They oscillate without any external input. Exogenous oscillators are modelled as contractive systems entrained by oscillatory inputs \cite{sontag2010contractive}. They cannot oscillate without external input. 
Endogenous oscillators are not reliable because autonomous limit cycles are phase sensitive \cite{Pikovsky2001}. Exogeneous oscillators do not resonate because they lack internal resonance mechanisms. 

Excitable systems borrow the best of the two worlds. They are exogenous, because they do not oscillate without external inputs. But they are also endogenous, because the mechanism of their oscillations is purely internal. The only role of the external input is to entrain the internal rhythm in a reliable manner.
	
	%

\section{Conclusions}
This paper has studied the contraction properties of conductance-based neuron models and their implications for the reliability of event timings. Contraction is a continuity property of the underlying continuous-time dynamics that supports the robustness of the associated discrete-event behavior. On the continuous side, we established incremental exponential stability of the unforced Hodgkin–Huxley dynamics on a compact forward-invariant set. Under impulsive forcing, we showed that contraction is preserved when impulses are sufficiently spaced in time, and derived an explicit sufficient condition for periodic inputs. At the opposite extreme, we characterized the high-rate limit: the synapse behaves as effectively always open, the dynamics approach a constant-conductance regime, and a limit cycle may emerge, thereby eroding the contraction margin. This clarifies the contrast between the robustness of forced contractive responses and the phase sensitivity inherent to autonomous oscillations.

On the discrete side, reliability follows from contraction of the continuous flow. With an event readout that yields isolated events, contraction makes trajectories converge after a transient, which in turn yields convergence of the detected spike times. The same mechanism also explains robustness across trials: bounded disturbances and small parameter variations produce uniformly bounded deviations in the trajectories and, consequently, uniformly bounded timing errors. Simulations with randomized initial conditions and parameter and kinetic perturbations illustrate both regimes: sparse inputs lead to a unique attractive steady response and reliable spike timing, whereas dense inputs induce tonic activity and a corresponding loss of timing reliability.

Beyond methodology, the results also clarify coding viewpoints. In contractive regimes, precise spike timing can serve as a robust carrier of information. In non-contractive regimes, by contrast, persistent phase variability undermines timing reproducibility, so that a rate-based description becomes the meaningful one.

Future work will extend the framework to broader classes of conductance-based neuron models with diverse ionic currents, incorporate richer synaptic dynamics and input classes, and address networks of excitable systems.
	
	\appendix
	
\begin{proof}[Proof of Lemma \ref{lemma.fi}]
	Consider $\vpre(t)$ as in \eqref{eq.disc.v_pre.dirac}. Between impulses the synapse dynamics obeys
	$\dot s=-\alpha s$, so one has
	\begin{equation}\label{thm.eq.syn.flow}
		s(t)=s(t_\ell^+)\,e^{-\alpha(t-t_\ell)},
		\quad \forall t\in[t_\ell,t_{\ell+1}).
	\end{equation}
	At an impulse time $t_\ell$, the jump condition is obtained by integrating
	\eqref{eq.synapse} across $t_\ell$. Since the decay term is integrable, only the impulsive term contributes to the jump. Setting $q\coloneq 1-s$,
	one gets across the impulse
	\[
	\dot q = -\beta q\,\delta(t-t_\ell),
	\]
	hence
	\[
	q(t_\ell^+) = q(t_\ell^-)e^{-\beta}.
	\]
	Equivalently,
	\[
	1-s(t_\ell^+) = \bigl(1-s(t_\ell^-)\bigr)e^{-\beta},
	\]
	that is,
	\begin{equation}\label{thm.eq.syn.jump}
		s(t_\ell^+) = 1-\bigl(1-s(t_\ell^-)\bigr)e^{-\beta}.
	\end{equation}
	Assume $s(t_\ell^-)\in[0,1]$. Then $1-s(t_\ell^-)\in[0,1]$, so
	\[
	0\le \bigl(1-s(t_\ell^-)\bigr)e^{-\beta}\le 1,
	\]
	and therefore $s(t_\ell^+)\in[0,1]$. Moreover, because of \eqref{thm.eq.syn.flow}, if $s(t_\ell^+)\in[0,1]$, then
	\[
	s(t)\in[0,1],\qquad \forall t\in[t_\ell,t_{\ell+1}).
	\]
	Therefore $\cS$ is forward invariant for \eqref{eq.synapse}.
	
	The gate dynamics evolve on $\cX$. Evaluating the vector field describing the dynamics at the boundaries $x_j\in\{0,1\}$ gives
	\begin{align*}
		x_j=0 \,:\, \tau_j(v)\dot x_j=\mu_j(v) &\implies \dot x_j>0,\\
		x_j=1 \,:\, \tau_j(v)\dot x_j=\mu_j(v)-1 &\implies \dot x_j<0,
	\end{align*}
	for all $j\in\{1,\dots,m\}$ and any $v$.
	For the voltage, at the boundaries $v\in\{E_{\min},E_{\max}\}$ we have
	\begin{align*}
		C\dot v = -G(x, s)\bigl(E_{\min}-E(x, s)\bigr) &\ge0,\\
		C\dot v = -G(x, s)\bigl(E_{\max}-E(x, s)\bigr) &\le0,
	\end{align*}
	for any $s \in \cS$, and $x=\col(x_1,\dots,x_m) \in \cX$. Therefore, $\cZ$ is forward invariant for \eqref{eq:exc_model}.
\end{proof}
	
	\begin{proof}[Proof of Lemma~\ref{lemma.les.attractive.es}]
		Consider system \eqref{eq.exc.compact} with $u\equiv 0$, and let $\varphi(t,z_0)$ be its solution starting from $z_0 \in \cZ$. By Assumption~\ref{ass.les.gas}, local exponential stability of $z^*$ implies the existence of $r>0$, $\lambda>0$, and $c\ge 1$ such that for all $z_0\in B_r(z^*)$, it holds that
		\begin{equation}\label{thm.eq.les}
			\|\varphi(t,z_0)-z^*\|\le c e^{-\lambda t}\|z_0-z^*\|,\ \forall t\ge 0,
		\end{equation}
		where
		\[
		B_r(z^*)\coloneq \{z\in\cZ:\|z-z^*\|<r\}.
		\]
		Since \eqref{eq.exc.compact} is time invariant, attractiveness of $z^*$ on the compact set $\cZ$ is uniform. Hence there exists $T^*>0$ such that
		\[
		\varphi(t,z_0)\in B_r(z^*),
		\quad \forall z_0\in\cZ,\ \forall t\ge T^*.
		\]
		Set
		\[
		R\coloneq \sup_{\xi\in\cZ}\|\xi-z^*\|<\infty.
		\]
		If $\|z_0-z^*\|\le r$, then \eqref{thm.eq.les} already gives
		\[
		\|\varphi(t,z_0)-z^*\|\le c e^{-\lambda t}\|z_0-z^*\|,
		\qquad \forall t\ge 0.
		\]
		
		If $\|z_0-z^*\|>r$, then for $t\in[0,T^*)$,
		\[
		\|\varphi(t,z_0)-z^*\|\le R
		\le \frac{R}{r}\|z_0-z^*\|
		\le \frac{R}{r}e^{\lambda T^*}e^{-\lambda t}\|z_0-z^*\|.
		\]
		For $t\ge T^*$, applying \eqref{thm.eq.les} from time $T^*$ yields
		\[
		\|\varphi(t,z_0)-z^*\|
		\le c e^{-\lambda (t-T^*)}\|\varphi(T^*,z_0)-z^*\|.
		\]
		Using $\|\varphi(T^*,z_0)-z^*\|\le R\le \frac{R}{r}\|z_0-z^*\|$, we get
		\[
		\|\varphi(t,z_0)-z^*\|
		\le c\frac{R}{r}e^{\lambda T^*}e^{-\lambda t}\|z_0-z^*\|.
		\]
		Therefore, with
		\[
		K\coloneq \max\left\{c,\frac{R}{r}e^{\lambda T^*},c\frac{R}{r}e^{\lambda T^*}\right\},
		\]
		we obtain, for all $z_0\in\cZ$,
		\[
		\|\varphi(t,z_0)-z^*\|\le K e^{-\lambda t}\|z_0-z^*\|,
		\qquad \forall t\ge 0.
		\]
		Thus $z^*$ is exponentially stable on $\cZ$.
	\end{proof}
	
\begin{proof}[Proof of Theorem~\ref{thm.contr.no.u}]
	Let \mbox{$\varphi(t,z_0)$} be a solution of \eqref{eq.exc.compact} with $u\equiv 0$ from $z_0\in\cZ$, and let $\D F(z)$ be its Jacobian. Local exponential stability of $z^*$ implies the existence of $P,Q\succ 0$ and $r>0$ such that
	\[
	\D F(z)^{\top}P + P \D F(z) \le -Q,\quad \forall z\in B^{\norm{\cdot}_P}_{r}(z^*),
	\]
	which implies contraction in the $P$-metric (see \cite[Thm 2.29]{pavlov_uniform_2006}). By norm equivalence, there exist $m,M>0$ such that $m\|\xi\|\le\|\xi\|_P\le M\|\xi\|$ for all $\xi$. Hence
	\[
	B_{r/M}(z^*)\subseteq B^{\|\cdot\|_P}_{r}(z^*)\subseteq B_{r/m}(z^*).
	\]
	Let $\bar{r}\coloneq r/M$. Then by converting the $P$-metric estimate to the Euclidean norm we obtain constants $\lambda>0$ and $c\ge 1$ such that
	\begin{equation}\label{thm.eq.local.contr}
		\|\varphi(t,x_0)-\varphi(t,y_0)\|\le c e^{-\lambda t}\|x_0-y_0\|,\quad \forall x_0,y_0\in B_{\bar{r}}(z^*).
	\end{equation}
	Uniform attractivity of $z^*$ on $\cZ$ yields $T^*>0$ such that $\varphi(t,z_0)\in B_{\bar r}(z^*)$ for all $t>T^*$ and all $z_0\in\cZ$. For arbitrary $x_0,y_0\in\cZ$, set $\Delta(t)\coloneq \varphi(t,x_0)-\varphi(t,y_0)$. Since $F$ is $C^1$ on $\cZ$, there exists $L\ge 0$ such that
	\[
	\|F(p)-F(q)\|\le L\|p-q\|,\quad \forall p,q\in\cZ.
	\]
	Thus, for $t\in[0,T^*]$, it holds
	\begin{align*}
		\frac{d}{dt}\tfrac12\|\Delta(t)\|^2
		&= \langle \Delta(t), F(\varphi(t,x_0)) - F(\varphi(t,y_0))\rangle\\
		& \le L\|\Delta(t)\|^2,
	\end{align*}
	hence by the \emph{Grönwall Lemma}
	\[
	\|\Delta(t)\|\le e^{Lt}\|x_0-y_0\|,\quad t\in[0,T^*].
	\]
	For $t\ge T^*$, \eqref{thm.eq.local.contr} gives
	\[
	\|\Delta(t)\|\le c e^{-\lambda(t-T^*)}\|\Delta(T^*)\|.
	\]
	Combining the two estimates, we obtain
	\begin{align*}
		\|\Delta(t)\|
		&\le c e^{-\lambda(t-T^*)} e^{L T^*}\|x_0-y_0\|\\
		&= c e^{\lambda T^*} e^{-\lambda t} e^{L T^*}\|x_0-y_0\|\\
		&= K e^{-\lambda t}\|x_0-y_0\|,
	\end{align*}
	for $t\ge 0$, and with $K\coloneq c e^{(L+\lambda)T^*}$. Then, the unforced system is IES on $\cZ$.
\end{proof}
	
\begin{proof}[Proof of Theorem~\ref{thm:IES-dwell}]
	Fix $t_2\ge t_1\ge0$ and two solutions $\varphi(\cdot,x_0,u)$, $\varphi(\cdot,y_0,u)$ under the same input $u$ given by \eqref{eq.disc.v_pre.dirac}. Let $\Sigma\subset\R{}_+$ be the set of impulse times, and list those lying in $(t_1,t_2]$ as
	\[
	\sigma_1<\dots<\sigma_N,
	\qquad
	N:=\card\{\sigma\in\Sigma:\ t_1<\sigma\le t_2\}\in\mathbb N_0.
	\]
	Define the partition
	\[
	\rho_0:=t_1,\qquad \rho_i:=\sigma_i\ (i=1,\dots,N),\qquad \rho_{N+1}:=t_2.
	\]
	Set
	\[
	\Delta(t)\coloneq \varphi(t,x_0,u)-\varphi(t,y_0,u).
	\]
	On each flow segment $[\rho_i,\rho_{i+1})$ there are no impulses, so by Theorem~\ref{thm.contr.no.u} we have
	\[
	\|\Delta(\rho_{i+1}^-)\|
	\le k\,e^{-\lambda(\rho_{i+1}-\rho_i)}\,\|\Delta(\rho_i^+)\|.
	\]
	At each impulse $\sigma_i=\rho_i$ only the synaptic variable $s$ jumps, according to
	\[
	s^+=1-(1-s^-)e^{-\beta}.
	\]
	Therefore, it holds that
	\[
	|s_1^+-s_2^+|
	=
	e^{-\beta}|s_1^--s_2^-|
	\le |s_1^--s_2^-|.
	\]
	Hence, the jump map is non-expansive, and thus
	\[
	\|\Delta(\rho_i^+)\|\le \|\Delta(\rho_i^-)\|.
	\]
	Chaining these $N$ jumps and $N+1$ flow segments yields
	\begin{align*}
		\|\Delta(t_2)\|
		=
		\|\Delta(\rho_{N+1}^-)\|
		&\le
		\Bigl(\prod_{i=0}^{N}k e^{-\lambda(\rho_{i+1}-\rho_i)}\Bigr)\|\Delta(t_1)\|\\
		&=
		k^{N+1}e^{-\lambda(t_2-t_1)}\|\Delta(t_1)\|.
	\end{align*}
	By the interval average dwell-time bound,
	\[
	N\le N_0+\frac{t_2-t_1}{\tau_a},
	\]
	hence
	\[
	k^{N+1}\le k^{N_0+1}e^{\frac{\ln k}{\tau_a}(t_2-t_1)}.
	\]
	Therefore
	\[
	\|\Delta(t_2)\|
	\le
	k^{N_0+1}
	e^{-\left(\lambda-\frac{\ln k}{\tau_a}\right)(t_2-t_1)}
	\|\Delta(t_1)\|,
	\quad \forall t_2\ge t_1\ge 0.
	\]
	If $\lambda>\frac{\ln k}{\tau_a}$, the exponential rate is positive, and the claim follows.
\end{proof}

\begin{proof}[Proof of Corollary~\ref{cor:periodic-IES}]
	Let $\Sigma=\{t_\ell=\ell T:\ell\in\mathbb N\}$. For any $t_2\ge t_1\ge0$ let
	\[
	N:=\card\{\sigma\in\Sigma: t_1<\sigma\le t_2\}.
	\]
	Since $t_{\ell+1}-t_\ell=T$, we have $N\le 1+\frac{t_2-t_1}{T}$. Hence, the average dwell-time condition of Theorem~\ref{thm:IES-dwell} holds with $\tau_a=T$ and $N_0=1$. Applying Theorem~\ref{thm:IES-dwell} gives IES provided $\lambda>\frac{\ln k}{\tau_a}=\frac{\ln k}{T}$, that is, $1/T<\lambda/\ln k$.
	\end{proof}
\begin{proof}[Proof of Lemma~\ref{lem:high-rate-synapse}]
	From the proof of Lemma~\ref{lemma.fi}, for $t\in[t_\ell,t_{\ell+1})$,
	\[
	s(t)=s(t_\ell^+)e^{-\alpha(t-t_\ell)},
	\]
	and at impulse times $t_\ell$,
	\[
	s(t_\ell^+)=1-\bigl(1-s(t_\ell^-)\bigr)e^{-\beta}.
	\]
	Since the impulses are $T$-periodic, $t_{\ell+1}-t_\ell=T$. Define
	\[
	s_\ell\coloneq s(t_\ell^+).
	\]
	Then, by the flow on $[t_\ell,t_{\ell+1})$,
	\[
	s(t_{\ell+1}^-)=e^{-\alpha T}s_\ell.
	\]
	Therefore
	\begin{align*}
		s_{\ell+1}
		&=1-\Bigl(1-s(t_{\ell+1}^-)\Bigr)e^{-\beta }\\
		&=1-\Bigl(1-e^{-\alpha T}s_\ell\Bigr)e^{-\beta }\\
		&=e^{-\beta }e^{-\alpha T}s_\ell + (1-e^{-\beta§})\\
		&=: a_T s_\ell + (1-e^{-\beta}),
	\end{align*}
	where
	\[
	a_T\coloneq e^{-\beta}e^{-\alpha T}\in(0,1).
	\]
	The affine map $s_{\ell+1}=a_T s_\ell+(1-e^{-\beta})$ is a strict contraction and admits the unique fixed point
	\[
	s_T^\star
	=
	\frac{1-e^{-\beta}}{1-a_T}
	=
	\frac{1-e^{-\beta}}{1-e^{-\beta}e^{-\alpha T}}.
	\]
	Hence $s_\ell\to s_T^\star$ exponentially for any $s_0\in[0,1]$. The corresponding $T$-periodic solution is
	\[
	s_*^{(T)}(t)=s_T^\star e^{-\alpha(t-\ell T)},
	\qquad t\in[\ell T,(\ell+1)T).
	\]
	Moreover, letting $\ell=\lfloor t/T\rfloor$, for $t\in[\ell T,(\ell+1)T)$,
	\begin{align*}
		|s(t)-s_*^{(T)}(t)|
		&=|s_\ell-s_T^\star| e^{-\alpha(t-\ell T)}\\
		&\le a_T^\ell e^{-\alpha(t-\ell T)}|s_0-s_T^\star|\\
		&\le a_T^{-1} e^{-\gamma_T t}|s_0-s_T^\star|,
	\end{align*}
	where
	\[
	\gamma_T:=-\frac{1}{T}\ln a_T>0.
	\]
	Thus $s_*^{(T)}$ is globally exponentially attractive.
	
	To prove convergence of $s_*^{(T)}$ to $1$, first note that for $t\in[\ell T,(\ell+1)T)$,
	\[
	s_*^{(T)}(t)\in[s_T^\star e^{-\alpha T},\,s_T^\star],
	\]
	so
	\[
	\sup_{t\ge 0}|s_*^{(T)}(t)-1|
	=
	\max\bigl\{1-s_T^\star,\ 1-s_T^\star e^{-\alpha T}\bigr\}.
	\]
	A direct calculation gives
	\begin{align*}
		1-s_T^\star
		&=
		\frac{e^{-\beta}(1-e^{-\alpha T})}{1-e^{-\beta}e^{-\alpha T}},\\
		1-s_T^\star e^{-\alpha T}
		&=
		\frac{1-e^{-\alpha T}}{1-e^{-\beta}e^{-\alpha T}}.
	\end{align*}
	Using $1-e^{-x}\le x$ for $x\ge 0$ and
	\[
	1-e^{-\beta}e^{-\alpha T}\ge 1-e^{-\beta},
	\]
	we obtain the uniform bound
	\begin{equation}\label{thm.unif.bound}
		\sup_{t\ge 0}|s_*^{(T)}(t)-1|
		\le
		\frac{\alpha T}{1-e^{-\beta}}
		\xrightarrow[T\to 0]{}0.
	\end{equation}
	Therefore, for every $L>0$ and every $\phi\in \mathcal L^1([0,L])$,
	\[
	\left|\int_0^L \phi(t)\bigl(s_*^{(T)}(t)-1\bigr)\,dt\right|
	\le
	\|\phi\|_{\mathcal L^1}\sup_{t\in[0,L]}|s_*^{(T)}(t)-1|,
	\]
	and hence, in view of \eqref{thm.unif.bound}, $s_*^{(T)}\to 1$ as $T\to 0$.
\end{proof}

	\bibliography{biblio}
	\bibliographystyle{IEEEtran}
	
\end{document}